\documentclass[nocopyrightspace]{sigplanconf}

\usepackage{amsmath}
\usepackage[]{graphicx}
\usepackage{latexsym}

\newcommand{\DDC}{D_{\rm DC}}
\newcommand{\TDC}{T_{\rm DC}}

\newcommand{\PDC}{P_{\rm DC}}

\newcommand{\TL}{{\it TL}}
\newcommand{\len}{{\it len}}
\newcommand{\lambdapar}{\lambda_{\it par}}

\newcommand{\pos}{{\it pos}}

\newcommand{\modelsa}{\mathrel{\models_{\rm MDF/DC}}}
\newcommand{\modelsb}{\mathrel{\models_{\rm DC}}}

\newcommand{\SibDF}{{\it SibDF}}
\newcommand{\DFS}{{\it DFS}}

\newcommand{\eval}{{\rm eval}}

\newtheorem{definition}{Definition}
\newtheorem{theorem}{Theorem}
\newtheorem{lemma}{Lemma}
\newtheorem{corollary}{Corollary}
\newtheorem{example}{Example}

\newcommand{\qed}{\mbox{}\hfill$\Box$}

\newenvironment{proofs}
  {\begin{trivlist}\item[]\emph{Proof Sketch.\ }}
  {\qed\end{trivlist}}
\newenvironment{proof}
  {\begin{trivlist}\item[]\emph{Proof.\ }}
  {\qed\end{trivlist}}

\begin{document}

\conferenceinfo{DBPL 2013}{date, City.} 
\copyrightyear{2013} 
\copyrightdata{[to be supplied]} 


\title{XPath Satisfiability with Parent Axes or
Qualifiers Is Tractable under Many of Real-World DTDs}

\authorinfo{Yasunori Ishihara}
           {Osaka University}
           {ishihara@ist.osaka-u.ac.jp}
\authorinfo{Nobutaka Suzuki}
           {University of Tsukuba}
           {nsuzuki@slis.tsukuba.ac.jp}
\authorinfo{Kenji Hashimoto}
           {Nara Institute of Science and Technology}
           {k-hasimt@is.naist.jp}
\authorinfo{Shogo Shimizu}
           {Gakushuin Women's College}
           {shogo.shimizu@gakushuin.ac.jp}
\authorinfo{Toru Fujiwara}
           {Osaka University}
           {fujiwara@ist.osaka-u.ac.jp}

\maketitle

\begin{abstract}
This paper aims at finding a subclass of DTDs that covers
many of the real-world DTDs while offering a polynomial-time
complexity for deciding the XPath satisfiability problem.
In our previous work, we proposed RW-DTDs,
which cover most of the real-world DTDs
(26 out of 27 real-world DTDs and 1406 out of 1407 DTD rules).
However, under RW-DTDs,
XPath satisfiability with only child,
descendant-or-self, and sibling axes is tractable.

In this paper, we propose MRW-DTDs,
which are slightly smaller than RW-DTDs but
have tractability on XPath satisfiability
with parent axes or qualifiers.
MRW-DTDs are a proper superclass of
duplicate-free DTDs proposed by Montazerian et al., and
cover 24 out of the 27 real-world DTDs and
1403 out of the 1407 DTD rules.
Under MRW-DTDs, we show that
XPath satisfiability problems with
(1)~child, parent, and sibling axes, and
(2)~child and sibling axes and qualifiers
are both tractable,
which are known to be intractable under RW-DTDs.
\end{abstract}




\section{Introduction}
\label{sec:Introduction}

XPath satisfiability is one of the major theoretical topics
in the field of XML databases.
XPath is a query language for XML documents,
where an XML document is often regarded
as an unranked labeled ordered tree.
An XPath expression specifies a pattern of (possibly branching) paths
from the root of a given XML document.
The answer to an XPath expression for an XML document $T$
is a set of nodes $v$ of $T$ such that the specified path pattern
matches the path from the root to $v$.
A given XPath expression $p$ is satisfiable under a given DTD
(Document Type Definition) $D$ if there is an XML document $T$
conforming to $D$ such that the answer to $p$ for $T$ is a nonempty
set.

One of the motivations for
research on XPath satisfiability is query optimization.
When (a part of) an XPath expression is found unsatisfiable,
we can always replace the expression with the empty set without
evaluating it.
Another motivation is to decide consistency and absolute consistency of
XML schema mappings~\cite{ABLM10,KHIF13},
which are desirable properties
for realizing XML data exchange and integration.
The decision problem of such properties can be reduced to
XPath satisfiability problem.

Unfortunately, it is known that satisfiability
under unrestricted DTDs is in P only for a
very small subclass of XPath expressions,
namely, XPath with only child axis, descendant-or-self axis,
and path union~\cite{BFG05,BFG08}.
To the best of our knowledge,
two approaches have been adopted so far
in order to resolve the intractability of XPath satisfiability.
The approach adopted by
Genev{\`e}s and Laya\"{\i}da 
is to translate XPath expressions to formulas
in monadic second-order (MSO) logic~\cite{GL06} and
in a variant of $\mu$-calculus~\cite{GL07,GLS07}.
Regular tree grammars~\cite{MLMK05},
which are a general model of XML schemas and a proper
superclass of DTDs, are also translated to such formulas.
Then, satisfiability is verified by fast decision procedures for
MSO and $\mu$-calculus formulas.
The other approach is
to find a tractable combination of XPath classes and DTD classes.
For example, Lakshmanan et al.\ examined
satisfiability under non-recursive DTDs~\cite{LRWZ04}, and
Benedikt et al.\ investigated
non-recursive and disjunction-free DTDs~\cite{BFG05,BFG08,GF05}.
However, non-recursiveness does not broaden the tractable
class of XPath.
Disjunction-freeness definitely broadens the tractable class of XPath,
but disjunction-free DTDs are too restricted from a practical
point of view.

\begin{table}[t]
\caption{The numbers of RW, MRW, DF, and DC$^{?+\#}$
rules in real-world DTDs.}
\label{tab:prac}
\medskip\par
\centering{
\begin{tabular}{l|r@{~~~}|r@{~~~~}r@{~~~~~}r@{~~~~}r@{~~~~~~~}}\hline
DTD Name  & \multicolumn{5}{c}{~Numbers of Rules~} \\\cline{2-6}
          & \multicolumn{1}{c|}{Total~} &
            \multicolumn{1}{c}{~~RW~} &
            \multicolumn{1}{c}{\textbf{MRW}} &
            \multicolumn{1}{c}{~~DF~~} &
            \multicolumn{1}{c}{DC$^{?+\#}$}
\\\hline
DBLP          &   36 & \bf   36  & \bf   36  & \bf   36 & \bf   36 \\
Ecoknowmics   &  224 & \bf  224  &      223  &      223 &      222 \\
LevelOne      &   28 & \bf   28  & \bf   28  & \bf   28 &       26 \\
MathML-2.0    &  181 & \bf  181  & \bf  181  & \bf  181 & \bf  181 \\
Mondial       &   40 & \bf   40  & \bf   40  & \bf   40 & \bf   40 \\
Music ML      &   12 & \bf   12  &       10  &       10 & \bf   12 \\
News ML       &  118 & \bf  118  & \bf  118  & \bf  118 &      114 \\
Newspaper     &    7 & \bf    7  & \bf    7  & \bf    7 & \bf    7 \\
Opml          &   15 & \bf   15  & \bf   15  & \bf   15 & \bf   15 \\
OSD           &   15 & \bf   15  & \bf   15  & \bf   15 &       14 \\
P3P-1.0       &   85 & \bf   85  & \bf   85  &       73 &       83 \\
PSD           &   66 & \bf   66  & \bf   66  & \bf   66 &       64 \\
Reed          &   16 & \bf   16  & \bf   16  & \bf   16 & \bf   16 \\
Rss           &   30 & \bf   30  & \bf   30  & \bf   30 &       29 \\
SigmodRecord  &   11 & \bf   11  & \bf   11  & \bf   11 & \bf   11 \\
SimpleDoc     &   49 & \bf   49  & \bf   49  & \bf   49 & \bf   49 \\
SSML-1.0      &   16 & \bf   16  & \bf   16  & \bf   16 & \bf   16 \\
SVG-1.1       &   80 & \bf   80  & \bf   80  &       78 &       77 \\
TV-Schedule   &   10 & \bf   10  & \bf   10  &        9 & \bf   10 \\
VoiceXML-2.0  &   62 & \bf   62  & \bf   62  & \bf   62 & \bf   62 \\
Xbel-1.0      &    9 & \bf    9  & \bf    9  & \bf    9 & \bf    9 \\
XHTML1-strict &   77 &       76  &       76  &       76 &       74 \\
XMark DTD     &   77 & \bf   77  & \bf   77  & \bf   77 &       76 \\
XML Schema    &   26 & \bf   26  & \bf   26  &       25 &       20 \\
XML Signature &   45 & \bf   45  & \bf   45  &       44 &       45 \\
XMLTV         &   40 & \bf   40  & \bf   40  & \bf   40 & \bf   40 \\
Yahoo         &   32 & \bf   32  & \bf   32  & \bf   32 & \bf   32 \\
\hline                                                      
Total         & 1407 &     1406  &     1403  &     1386 & 1380 \\\hline
\end{tabular}
}
\end{table}

There are two successful results of the latter approach.
The first one is
\emph{duplicate-free DTDs}~\cite{MWM07},
\emph{DF-DTDs} for short,
proposed by Montazerian et al.
A DTD is duplicate-free if
every tag name appears at most once in each content model
(i.e., the body of each DTD rule).
Table~\ref{tab:prac} shows an empirical survey of
real-world DTDs.
Many of the DTDs are selected according to the examination by
Montazerian et al.~\cite{MWM07}, and
several practical DTDs such as MathML and SVG
are included in the examined DTDs.
As shown in the table,
1386 out of 1407 real-world DTD rules
are duplicate-free.
Montazerian et al. also showed that
satisfiability of XPath expressions with child axis and
qualifiers is tractable~\cite{MWM07}.
Later, other several tractable XPath classes
were presented in our previous work~\cite{SF09}.
The tractability mainly stems from easiness of analyzing
\emph{non-cooccurrence among tag names}.
More formally, a subexpression $e|e'$ of a content model
specifies non-cooccurrence between the tag names in $e$ and those of
$e'$.
In DF-DTDs,
each tag name can appear at most once in the content model,
so complicated non-cooccurrence among tag names is not expressible.

The other successful result is
\emph{disjunction-capsuled DTDs}~\cite{IMSHF09},
\emph{DC-DTDs} for short, and their extension
\emph{DC$^{?+\#}$-DTDs}~\cite{IHSF12}.
A DTD is disjunction-capsuled if in each content model,
every disjunction operator appears within a scope of a Kleene star
operator.
For example, $a(b|c)^*$ is DC but $(a|b)c^*$ is not.
XPath expressions were supposed to consist of
$\downarrow$ (child axis),
$\downarrow^*$ (descendant-or-self axis),
$\uparrow$ (parent axis),
$\uparrow^*$ (ancestor-or-self axis),
$\rightarrow^+$ (following-sibling axis),
$\leftarrow^+$ (preceding-sibling axis),
$\cup$ (path union), and
$[~]$ (qualifier).
Then, it was shown that the satisfiability under DC-DTDs for
XPath expressions without upward axes or qualifiers is tractable.
The tractability is mainly from the fact that
in DC-DTDs, any non-cooccurrence of tag names is abolished by
the surrounding Kleene star operator.
DC-DTDs were extended to \emph{DC$^{?+}$-DTDs}~\cite{ISF10}
by allowing
operators ``$?$'' (zero or one occurrence) and
``$+$'' (one or more occurrences) in a restricted manner,
and then,
to \emph{DC$^{?+\#}$-DTDs}~\cite{IHSF12} by allowing
a new operator $\#$ representing ``either or both.''
Precisely, $\#$ is an $(m+l)$-ary operator and
$(a_1,\ldots,a_m)\#(b_1,\ldots,b_l)$ is equivalent to
$a_1\cdots a_mb_1^?\cdots b_l^?|a_1^?\cdots a_m^?b_1\cdots b_l$.
Especially, $a\#b$ is equivalent to $a|b|ab$,
so it means ``either or both of $a$ and $b$.''
As shown in~Table~\ref{tab:prac},
1380 out of 1407 real-world DTD rules
are DC$^{?+\#}$.
Amazingly, \emph{all} the tractability of DC-DTDs is inherited by
DC$^{?+\#}$-DTDs~\cite{ISF10,IHSF12}.

Although more than 98\% of real-world DTD rules are
DF or DC$^{?+\#}$,
the ratio of DF-DTDs or DC$^{?+\#}$-DTDs is not so high.
Table~\ref{tab:prac} shows that
8 out of the 27 DTDs are not DF, 
12 are not DC$^{?+\#}$,
and 6 are neither DF nor DC$^{?+\#}$.
To overcome this weakness,
we proposed \emph{RW-DTDs}~\cite{IHSF12},
which are a proper superclass of
both DF-DTDs and DC$^{?+\#}$-DTDs.
To be specific,
RW-DTDs are not just the union of them,
but a ``hybrid'' class of them.
In each content model $e=e_1\cdots e_n$ of an RW-DTD,
each subexpression $e_i$ is either DC$^{?+\#}$
or DF in the whole content model.
For example, $a^*(b|c)a^*$ is neither DF nor DC$^{?+\#}$,
but is RW because
the non-DC$^{?+\#}$ part $(b|c)$ is DF in the whole content model.
On the other hand,
$a^*(b|c)b^*$ is not RW because
the non-DC$^{?+\#}$ part $(b|c)$ contains $b$,
which appears twice in the whole content model.
RW-DTDs cover
26 out of the 27 real-world DTDs, 1406 out of the 1407 DTD rules
(see Table~\ref{tab:prac} again).
However, RW-DTDs do not inherit all the tractability of
the original DTD classes.
Actually,
XPath satisfiability with only child,
descendant-or-self, and sibling axes is tractable
under RW-DTDs.

This paper aims at finding a large subclass of RW-DTDs
under which XPath satisfiability becomes tractable
for a broader class of XPath expressions.
The source of the
intractability of XPath satisfiability under RW-DTDs seemed
tag name occurrence of some fixed, plural number of times~\cite{IHSF12}.
According to this observation,
in this paper we propose \emph{MRW-DTDs},
which are RW-DTDs 
such that in each content model,
each symbol appears in the scope of a repetitive operator
(i.e., $*$ or $+$)
or DF in the whole content model.
For example, $a^*ba^*$ is MRW,
but $a^*ba$ is not MRW (although it is RW)
because the rightmost $a$ is not in the scope of
any repetitive operators
or DF in the whole content model.
MRW-DTDs are still a proper superclass of DF-DTDs but
incomparable to DC$^{?+\#}$-DTDs
(see Figure~\ref{fig:DTD classes}).
MRW-DTDs cover
24 out of the 27 real-world DTDs, 1403 out of the 1407 DTD rules
(see Table~\ref{tab:prac} again).

\begin{figure}[t]
\begin{center}
\includegraphics[scale=.45]{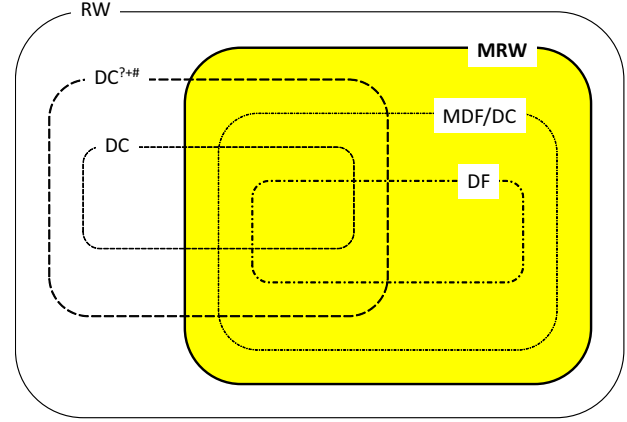}
\end{center}
\caption{Relationship among DTD classes.}
\label{fig:DTD classes}
\end{figure}

\begin{table*}[t]
\caption{Results of this paper and related works.}
\label{tab:results}
\medskip\par
\begin{center}
\newlength{\templength}
\settowidth{\templength}{$\rightarrow$}
\begin{tabular}{|c|c|c|c|c|c|c|c|c||c|c|c|c|}
\hline
\makebox[\templength][c]{$\downarrow$} &
\makebox[\templength][c]{$\downarrow^*$} &
\makebox[\templength][c]{$\uparrow$} &
\makebox[\templength][c]{$\uparrow^*$} &
\makebox[\templength][c]{$\rightarrow^+$} &
\makebox[\templength][c]{$\leftarrow^+$} &
\makebox[\templength][c]{$\cup$} &
\makebox[\templength][c]{$[~]_\wedge$} &
\makebox[\templength][c]{$[~]$} &
\begin{tabular}{c}~RW-DTDs~\end{tabular} &
\begin{tabular}{c}\textbf{MRW-DTDs}\end{tabular} &
\begin{tabular}{c}~~DF-DTDs~~\end{tabular} &
\begin{tabular}{@{~}c@{~}}DC$^{?+\#}$-DTDs\end{tabular}\\\hline\hline
+&+& & & & &+& & & P\cite{BFG05,BFG08}
                                 & P\cite{BFG05,BFG08}
                                    & P\cite{BFG05,BFG08} 
                                       & P\cite{BFG05,BFG08} 
\\\hline
+&+& & &+&+& & & & P\cite{IHSF12}  & P\cite{IHSF12}
                                    & P\cite{SF09} & P\cite{ISF10,IHSF12} 
\\\hline
+& &+& & & & & & & NPC\cite{IHSF12}
                                 & \textbf{P}
                                    & P\cite{SF09} & P\cite{ISF10,IHSF12} 
\\\hline
+& &+& &+&+& & & & NPC\cite{IHSF12}
                                 & \textbf{P}
                                    & P\cite{SF09} & P\cite{ISF10,IHSF12} 
\\\hline
+& & & & & & &+& & NPC\cite{IHSF12}
                                 & \textbf{P}
                                    & \textbf{P} & P\cite{ISF10,IHSF12} 
\\\hline
+& & & &+&+& &+& & NPC\cite{IHSF12}
                                 & \textbf{P}
                                    & \textbf{P} & P\cite{ISF10,IHSF12} 
\\\hline
+& & & & & &+&+&+& NPC\cite{BFG08,MWM07}
                                 & NPC\cite{BFG08,MWM07}
                                    & NPC\cite{BFG08,MWM07}
                                           & P\cite{ISF10,IHSF12} 
\\\hline
 &+& & & & & &+& & NPC\cite{BFG08,MWM07}
                                 & NPC\cite{BFG08,MWM07}
                                    & NPC\cite{BFG08,MWM07}
                                           & P\cite{ISF10,IHSF12} 
\\\hline
+&+&+& & & & & & & NPC\cite{SF09}  & NPC\cite{SF09}
                                    & NPC\cite{SF09} & P\cite{ISF10,IHSF12} 
\\\hline
+&+& & &+&+&+&+&+& NPC\cite{GF05,MWM07}
                                 & NPC\cite{GF05,MWM07}
                                    & NPC\cite{GF05,MWM07}
                                           & P\cite{ISF10,IHSF12} 
\\\hline
+&+&+&+&+&+&+& & & NPC\cite{GF05,MWM07}
                                 & NPC\cite{GF05,MWM07}
                                    & NPC\cite{GF05,MWM07}
                                           & P\cite{ISF10,IHSF12} 
\\\hline
+&+&+&+&+&+&+&+&+& NPC\cite{GF05,MWM07}
                                 & NPC\cite{GF05,MWM07}
                                    & NPC\cite{GF05,MWM07}
                                           & NPC\cite{GF05,ISF10} 
\\\hline
\end{tabular}
\medskip\par
NPC stands for NP-complete.
Bold letters indicate the contributions of this paper.
\end{center}
\end{table*}

Next, this paper shows that under MRW-DTDs,
XPath satisfiability problems with
(1)~child, parent, and sibling axes, and
(2)~child and sibling axes and qualifiers without disjunction
(denoted $[~]_\wedge$)
are both tractable.
Table~\ref{tab:results} summarizes
the results of this paper and related works.
Note that under RW-DTDs, satisfiability for child axes with
either parent axes or qualifiers 
is known to be NP-complete~\cite{IHSF12}.
Similarly to the case of RW-DTDs,
the decision algorithm for XPath satisfiability
under MRW-DTDs consists of the following two checks:
(1)~Check the satisfiability of a given XPath expression
under the DTD obtained by replacing
each disjunction with concatenation in a given MRW-DTD.
In other words, satisfiability is analyzed
as if the given MRW-DTD did not specify any
non-cooccurrence of tag names; and
(2)~Check that the given XPath expression does not violate
the non-cooccurrence specified by the original MRW-DTD.
The first check can be done by the efficient algorithm for
XPath satisfiability under DC-DTDs~\cite{IMSHF09,ISF10}.
To perform the second check,
we have to keep track of sets of already-traversed
sibling tag names and associate
the sets with nodes of a tree structure.
Since each tag name can appear at most once or
unboundedly many times in MRW-DTDs,
association of the sets to a tree structure 
is uniquely determined.
That enables us an efficient satisfiability checking.


The rest of this paper is organized as follows.
In Section~\ref{sec:Preliminaries}
several preliminary definitions to formalize the XPath satisfiability
problem are provided.
In Section~\ref{sec:Modeling Many of Real-World DTDs}
MRW-DTDs are proposed.
The tractability results under MRW-DTDs are
presented in Section~\ref{sec:Tractability results}.
Section~\ref{sec:Conclusions} summarizes the paper.

\section{Preliminaries}
\label{sec:Preliminaries}

\subsection{XML documents}
An XML document is represented by an unranked labeled ordered tree.
The label of a node $v$, denoted $\lambda(v)$, corresponds to a tag name.
We extend $\lambda$ to a function on sequences, i.e.,
for a sequence $v_1\cdots v_n$ of nodes,
let $\lambda(v_1\cdots v_n)=\lambda(v_1)\cdots\lambda(v_n)$.
A tree is sometimes denoted by a term, e.g.,
$a(b()c())$ denotes a tree consisting of three nodes;
the root has label $a$, and its left and right children have labels
$b$ and $c$, respectively.
Attributes are not handled in this paper.

\subsection{DTDs}

A regular expression over an alphabet $\Sigma$ consists of
constants $\epsilon$~(empty sequence) and the symbols in $\Sigma$,
and operators $\cdot$~(concatenation),
$*$~(repetition), 
$|$~(disjunction),
$?$~(zero or one occurrence), $+$~(one or more occurrences),
and $\#$ (either or both).
Here, $\#$ is an $(m+l)$-ary operator and
$(a_1,\ldots,a_m)\#(b_1,\ldots,b_l)$ is equivalent to $a_1\cdots
a_mb_1^?\cdots b_l^?|a_1^?\cdots a_m^?b_1\cdots b_l$.  
We exclude $\emptyset$~(empty set)
because we are interested in only nonempty regular languages.
The concatenation operator is often omitted as usual.
The string language represented by a regular expression $e$ is
denoted by $L(e)$.

A regular expression $e$ is \emph{duplicate-free}~\cite{MWM07} 
(\emph{DF} for short) if
every symbol in $e$ appears only once.
On the other hand, a regular expression $e$ is 
\emph{DC$^{?+\#}$}~\cite{IHSF12} if
$e$ is in the form of $e_1e_2\cdots e_n$ ($n\geq 1$),
where each $e_i$ ($1\leq i\leq n$) is either

\begin{itemize}
\item a symbol in $\Sigma$,
\item in the form of $(e_i')^*$ for a regular expression $e_i'$,
\item in the form of $(e_i')^?$ for a DC$^{?+\#}$ regular expression
  $e_i'$,
\item in the form of $(e_i')^+$ for a regular expression $e_i'$, or
\item in the form of
  $(e_{i1}',\ldots,e_{im}')\#(e_{i1}'',\ldots,e_{il}'')$ for
  DC$^{?+\#}$ regular expressions $e_{i1}',\ldots,e_{im}'$,
  $e_{i1}'',\ldots,e_{il}''$,
\end{itemize}
DC$^{?+\#}$ regular expressions are intended to exclude any
non-cooccurrence among symbols.
The argument of operators $*$ and $+$ can be an arbitrary regular
expression.
Such operators can abolish any non-cooccurrence specified
by their argument because the operators can repeat any subexpression
of their argument arbitrary times.
On the other hand, the argument of operators $?$ and $\#$
must be a
DC$^{?+\#}$ regular expression because the operators cannot repeat their
argument.
A DC$^{?+\#}$ regular expression $e$ is
\emph{disjunction-capsuled}~\cite{IMSHF09} (\emph{DC} for short) if
$e$ does not contain $?$, $+$, or $\#$.

The \emph{length} of a DC regular expression
$e=e_1e_2\cdots e_n$
is defined as the number $n$ of subexpressions of the top-level
concatenation operator, and denoted by $\len(e)$.
Moreover, $i$ ($1\leq i\leq\len(e)$) is called a $\emph{position}$ and
each $e_i$ is called the $i$-th subexpression of $e$.


\begin{definition}
A \emph{DTD} is a triple $D=(\Sigma,r,P)$, where
\begin{itemize}
\item
$\Sigma$ is a finite set of \emph{labels},
\item
$r\in\Sigma$ is the \emph{root label}, and
\item
$P$ is a mapping from $\Sigma$ to the set of regular expressions over
$\Sigma$.
Regular expression $P(a)$ is called the \emph{content model} of label $a$.
\end{itemize}
A \emph{duplicate-free DTD} (\emph{DF-DTD} for short) is a DTD
such that $P(a)$ is DF for every $a\in\Sigma$.
A \emph{disjunction-capsuled DTD} (\emph{DC-DTD} for short), is a DTD
such that $P(a)$ is DC for every $a\in\Sigma$.
A \emph{DC$^{?+\#}$-DTD} is a DTD
such that $P(a)$ is DC$^{?+\#}$ for every $a\in\Sigma$.
\end{definition}


\begin{definition}
A tree $T$ \emph{conforms} to a DTD $D=(\Sigma,r,P)$ if
\begin{itemize}
\item
the label of the root of $T$ is $r$, and
\item
for each node $v$ of $T$ and its children sequence
$v_1\cdots v_n$,
$L(P(\lambda(v)))$ contains
$\lambda(v_1\cdots v_n)$.
\end{itemize}
Let $\TL(D)$ denote the set of all the trees conforming to $D$.
\end{definition}

In this paper, we assume that every DTD $D=(\Sigma,r,P)$ contains
no useless symbols.
That is, for each $a\in\Sigma$, there is a tree $T$ conforming to $D$
such that the label of some node of $T$ is $a$.

The size of a regular expression is the number of constants and
operators appearing in the regular expression.
The size of a DTD is the sum of the sizes of all content models.

\subsection{XPath expressions}
\label{sec:XPath expressions}

The syntax of an XPath expression $p$ is defined as follows:
\begin{eqnarray*}
p & ::= & {\chi::l}\mid{p/p}\mid{p\cup p}\mid{p[q]},
\\
\chi & ::= & {\downarrow}\mid{\uparrow}\mid{\downarrow^*}
\mid{\uparrow^*}\mid{\rightarrow^+}\mid{\leftarrow^+},
\\
q & ::= & {p}\mid{q\wedge q}\mid{q\vee q},
\end{eqnarray*}
where $l\in\Sigma$.
Each $\chi\in\{{\downarrow},{\uparrow},{\downarrow^*},
{\uparrow^*},{\rightarrow^+},{\leftarrow^+}\}$
is called an \emph{axis}.
Also, a subexpression in the form of $[q]$ is called a
\emph{qualifier}.
An expression in the form of $\chi::l$ is said to be \emph{atomic}.
The \emph{size} of an XPath expression $p$ is defined as the number of
atomic subexpressions in $p$. 

The semantics of an XPath expression over a tree $T$
is defined as follows,
where $p$ and $q$ are regarded as binary and unary
predicates on paths from the root node of $T$, respectively.
In what follows, $v_0$ denotes the root of $T$, and
$v$ and $v'$ denote nodes of $T$.
Also, $w$, $w'$, and $w''$ are
nonempty sequences of nodes of $T$ starting by $v_0$,
unless otherwise stated.
\begin{itemize}
\item
$T\models(\downarrow::l)(w,wv')$
if path $wv'$ exists in $T$ and $\lambda(v')=l$.
\item
$T\models(\uparrow::l)(wv,w)$
if path $wv$ exists in $T$
and the label of the last node of $w$ is $l$.
\item
$T\models(\downarrow^*::l)(w,ww')$
if path $ww'$ exists in $T$ and the label of the last node of $ww'$
is $l$,
where $w'$ is a possibly empty sequence of nodes of $T$.
\item
$T\models(\uparrow^*::l)(ww',w)$
if path $ww'$ exists in $T$ and 
the label of the last node of $w$ is $l$,
where $w'$ is a possibly empty sequence of nodes of $T$.
\item
$T\models(\rightarrow^+::l)(wv,wv')$
if paths $wv$ and $wv'$ exist in $T$,
$v'$ is a following sibling of $v$, and $\lambda(v')=l$.
\item
$T\models(\leftarrow^+::l)(wv,wv')$
if paths $wv$ and $wv'$ exist in $T$,
$v'$ is a preceding sibling of $v$, and $\lambda(v')=l$.
\item
$T\models(p/p')(w,w')$
if there is $w''$ such that
$T\models p(w,w'')$ and
$T\models p'(w'',w')$.
\item
$T\models(p\cup p')(w,w')$
if $T\models p(w,w')$ or
$T\models p'(w,w')$.
\item
$T\models(p[q])(w,w')$
if $T\models p(w,w')$ and
$T\models q(w')$.
\item
$T\models p(w)$
if there is $w'$ such that
$T\models p(w,w')$.
\item
$T\models(q\wedge q')(w)$
if $T\models q(w)$ and
$T\models q'(w)$.
\item
$T\models(q\vee q')(w)$
if $T\models q(w)$ or
$T\models q'(w)$.
\end{itemize}

A tree $T$ \emph{satisfies} an XPath expression $p$ if
there is a node $v$ such that $T\models p(v_0,v)$,
where $v_0$ is the root node of $T$.
An XPath expression $p$ is \emph{satisfiable} under a DTD $D$
if some $T\in\TL(D)$ satisfies $p$.

In this paper, we often consider
\emph{qualifiers without disjunction}.
In this case the syntax of $p$ is simply redefined as
\begin{eqnarray*}
p & ::= & {\chi::l}\mid{p/p}\mid{p\cup p}\mid{p[p]}.
\end{eqnarray*}
Note that conjunction can be represented by a sequence of
qualifiers (e.g., $p[p'\wedge p'']$ can be represented
by $p[p'][p'']$).

Following the notation of~\cite{BFG05,BFG08},
a subclass of XPath is indicated by $\mathcal{X}(\cdot)$.
For example, the subclass with child axes and qualifiers
without disjunction is
denoted by $\mathcal{X}({\downarrow},[~]_\wedge)$.



\section{Modeling Many of Real-World DTDs}
\label{sec:Modeling Many of Real-World DTDs}

In this section, we introduce MRW-DTDs, which are a subclass of
RW-DTDs~\cite{IHSF12}.

RW-DTDs are defined as a hybrid class of DF and
DC$^{?+\#}$-DTDs. Formally, a regular expression $e$ is
\emph{RW} if $e$ is in the form of $e_1e_2\cdots e_n$
($n\geq 1$), where each $e_i$ ($1\leq i\leq n$) is either
\begin{itemize}
\item DC$^{?+\#}$; or
\item a regular expression consisting of only symbols from $\Sigma$
appearing once in $e$.
\end{itemize}
A DTD $D$ is called an \emph{RW-DTD} if
each content model of $D$ is RW.

Although RW-DTDs cover most of real-world DTDs, it is shown that
satisfiabilities of $\mathcal{X}({\downarrow},{\uparrow})$ and
$\mathcal{X}({\downarrow},[~]_\wedge)$ under RW-DTDs are both
NP-complete~\cite{IHSF12}.
This intractability is caused by non-repetitive symbols
(i.e., appearing outside the scope of any $*$ and $+$ operators) in a
DC$^{?+\#}$ part, which raise a combinatorial explosion. To handle
this problem, we define a slightly restricted version of RW-DTDs,
denoted \emph{MRW-DTDs}, in which non-repetitive symbols must appear
at most once in each context model.

\begin{definition}
  An RW-DTD $D = (\Sigma,r,P)$ is called an \emph{MRW-DTD} if for each
  content model $e$ and each symbol $a$ appearing in $e$, $a$
  appears once in $e$ whenever $a$ is outside the scope of any $*$ and $+$.
\end{definition}

\begin{example}
  Let $D = ( \{r,a,b,c\},r,P)$ be a DTD, where 
  \begin{eqnarray*}
    P(r) & = & (a|b)^*ca^+,\\
    P(a) = P(b) = P(c)  & = &\epsilon.  
  \end{eqnarray*}
  Then $D$ is an MRW-DTD.  On the other hand, consider a DTD $D' = (
  \{r,a,b,c\},r,P')$, where
  \begin{eqnarray*}
    P'(r) & = & (a|b)^*ca^?,\\
    P'(a) = P'(b) = P'(c)  & = &\epsilon.  
  \end{eqnarray*}
  Then $D'$ is RW but not MRW since in $P'(r)$ symbol $a$
  occurs twice but one of them appears outside the scope of $*$ and $+$.
\end{example}

We examined 27 real-world DTDs, 1407 rules
(see Table~\ref{tab:prac}).
During the examination, we found 6 DTD rules
which are not syntactically MRW
but can be transformed into equivalent MRW rules.
Specifically, the content models of the rules have the
following forms:
\begin{itemize}
\item
$ab^+|ab^+c$
(1 rule in Music ML),
\item
$a^*(bc^?d^?a^*|cd^?a^*|da^*)^?$
(1 rule in P3P-1.0),
\item
$a^*b^?(cdef^+|gc^?d^?e^?f^*)a^*$
(1 rule in P3P-1.0),
\item
$a(bc)^*|(bc)^+a((bc)^*)^?$
(2 rules in SVG-1.1), and
\item
$ab^?|b$
(1 rule in XML Signature).
\end{itemize}
These forms are equivalent to
\begin{itemize}
\item
$ab^+c^?$,
\item
$a^*((b\#(c\#d))a^*)^?$,
\item
$a^*b^?(g\#(c,d,e,f^+))a^*$,
\item
$(bc)^*a(bc)^*$, and
\item
$a\#b$,
\end{itemize}
respectively.
Therefore, we counted these 6 original DTD rules as MRW.

In summary, 24 out of the 27 real-world DTDs,
1403 out of the 1407 DTD rules were MRW.
Table~\ref{tab:form} shows 
the forms of the content models of the 4 DTD rules that are
not MRW.
Note that the form (F3) is not even RW.
Moreover, Music ML itself is a DC$^{?+\#}$-DTD
and therefore it is tractable for a broader class of XPath
expressions than that for MRW-DTDs.

\begin{table}[t]
\caption{The forms of the content models of the 4 rules that are
not MRW.}
\label{tab:form}
\begin{center}
\begin{tabular}{l|l}
\hline
(F1) &
$a^?b^?b^?c$
\quad(1 rule in Ecoknowmics)
\\
(F2) &
$a|aa$
\quad (2 rules in Music ML)
\\
(F3) &
$(a|b)^* ((c (a|b)^* (d (a|b)^*)^?) | (d (a|b)^* c (a|b)^*))$
\\
& \mbox{\qquad}(1 rule in XHTML1-strict)
\\
\hline
\end{tabular}
\end{center}
\end{table}


\section{Tractability Results under MRW-DTDs}
\label{sec:Tractability results}

We say that a regular expression $e$ is \emph{MDF/DC} if
$e$ is MRW but includes none of $?$, $+$, and $\#$.
An MRW-DTD is \emph{MDF/DC} if each content model is MDF/DC.
In this section, 
we first show that tractability of XPath satisfiability for MRW-DTDs is
identical to that for MDF/DC-DTDs,
if the XPath class is a subclass of
$\mathcal{X}({\downarrow},{\downarrow^\ast},{\uparrow},{\uparrow^\ast},
{\rightarrow^+},$ ${\leftarrow^+}$, ${\cup},{[\ ]})$.

Next, we provide a necessary and sufficient condition for
satisfiability of XPath expressions in
$\mathcal{X}({\downarrow},{\downarrow^\ast},{\uparrow},{\uparrow^\ast},
{\rightarrow^+},$ ${\leftarrow^+}$, ${\cup},{[\ ]})$
under MDF/DC-DTDs.
Similarly to our previous work~\cite{IMSHF09,ISF10,IHSF12},
we introduce a \emph{schema graph} of a given
MDF/DC-DTD, which represents parent-child relationship as well as the 
possible positions of the children specified by the MDF/DC-DTD.
Then we define a satisfaction relation between 
schema graphs and XPath expressions.
We show that the satisfaction relation coincides with
the satisfiability under MDF/DC-DTDs.

After that, we propose efficient algorithms for deciding
the satisfaction relation for two cases, namely,
$p\in\mathcal{X}({\downarrow},{\uparrow},{\rightarrow^+},
{\leftarrow^+})$ and
$p\in\mathcal{X}({\downarrow},{\rightarrow^+},
{\leftarrow^+},[~]_\wedge)$.
The decision algorithms consist
of the following two checks:
(1)~Check the satisfiability of $p$
under the DC-DTD obtained by replacing
disjunction with concatenation in a given MDF/DC-DTD; and
(2)~Check that $p$ does not violate
the non-cooccurrence specified by the original MDF/DC-DTD.
Actually, the satisfaction relation is defined
so that both of the checks can be done simultaneously.

\subsection{Tractability identicalness between MRW-DTDs and MDF/DC-DTDs}

First, let us review
\emph{satisfiability preservation relation}~$\sim$
discussed in~\cite{ISF10}.
Let $e$ and $e'$ be regular expressions.
We write $e\sim e'$ if they satisfy the 
following two conditions: 
\begin{itemize}
\item
every $w\in L(e)$
is a subsequence
(i.e., can be obtained by deleting zero or more symbols)
of some $w'\in L(e')$; and
\item
every $w'\in L(e')$
is a subsequence of
some $w\in L(e)$.
\end{itemize}
Let $D=(\Sigma,r,P)$ and $D'=(\Sigma,r,P')$.
We write $D\sim D'$ if $P(a)\sim P'(a)$ for each $a\in\Sigma$.
Since DTDs are assumed to have no useless symbols,
$D\sim D'$ implies that
\begin{itemize}
\item
every $T\in \TL(D)$
can be obtained by deleting zero or more subtrees of
some $T'\in \TL(D')$; and
\item
every $T'\in \TL(D')$
can be obtained by deleting zero or more subtrees of
some $T\in \TL(D)$.
\end{itemize}
Let $p\in\mathcal{X}({\downarrow},{\downarrow^\ast},{\uparrow},{\uparrow^\ast},
{\rightarrow^+},$ ${\leftarrow^+}$, ${\cup},{[\ ]})$ and
suppose that $D\sim D'$.
Then, $p$ is satisfiable under $D$ if and only if
$p$ is satisfiable under $D'$,
because our XPath class is positive (i.e., does not contain
negation operator) and not sensitive to next siblings
(i.e., cannot detect existence of nodes between two sibling nodes).
Thus, we have the following theorem:
\begin{theorem}[\cite{ISF10}]
\label{th:sim}
Suppose that classes $C$ and $C'$ of DTDs satisfy the following property:
for each DTD $D'\in C'$,
there exists $D\in C$ such that $D\sim D'$
and $D$ can be computed efficiently from $D'$.
Then, for any subclass $X$ of
$\mathcal{X}({\downarrow},{\downarrow^\ast},{\uparrow},{\uparrow^\ast},
{\rightarrow^+},$ ${\leftarrow^+}$, ${\cup},{[\ ]})$,
if the satisfiability problem for $X$ under $C$ is in P,
the same problem under $C'$ is also in P.
\end{theorem}

To apply Theorem~\ref{th:sim} to MRW-DTDs,
we introduce the following mapping $\delta$:
\begin{itemize}
\item
$\delta(\epsilon)=\epsilon$,
\item
$\delta(a)=a$ for each $a\in\Sigma$,
\item
$\delta(e_1\cdot e_2)=
\delta(e_1)\cdot\delta(e_2)$,
\item
$\delta(e^*)=(\delta(e))^*$,
\item
$\delta(e_1|e_2)=\delta(e_1)|\delta(e_2)$,
\item
$\delta(e^?)=\delta(e)$,
\item
$\delta(e^+)=(\delta(e))^*$, and
\item
$\delta((e_{11},\ldots,e_{1m})\#(e_{21},\ldots,e_{2l}))$\\
\qquad
$=\delta(e_{11})\cdots\delta(e_{1m})\cdot
\delta(e_{21})\cdots\delta(e_{2l})$.
\end{itemize}
Intuitively, $\delta$
removes all the $?$ operators, and
replaces all the $+$ and $\#$ operators 
with $*$ and $\cdot$ operators, respectively.
For example, $\delta(a^*((b\#(c\#d))a^*)^?)=a^*bcda^*$.
The next lemma is almost immediate:
\begin{lemma}
For any content model $e$ of an MRW-DTD,
$\delta(e)$ is MDF/DC.
\end{lemma}

Moreover, $\delta$ preserves satisfiability:
\begin{lemma}[\cite{IHSF12}]
$e\sim\delta(e)$ for any regular expression $e$. 
\end{lemma}

For a DTD $D=(\Sigma,r,P)$,
let $\delta(D)$ denote the DTD $(\Sigma,r,\delta(P))$, where
$\delta(P)(a)=\delta(P(a))$ for each $a\in\Sigma$.
By the above lemmas, we have $D\sim\delta(D)$,
and obviously $\delta(D)$ can be computed efficiently from $D$.
Moreover, $\delta(D)$ is MDF/DC if $D$ is an MRW-DTD.
Hence, from Theorem~\ref{th:sim} and the fact that
MDF/DC-DTDs are a subclass of MRW-DTDs, 
we have the following corollary:
\begin{corollary}
For any subclass $X$ of
$\mathcal{X}({\downarrow},{\downarrow^\ast},{\uparrow},{\uparrow^\ast},
{\rightarrow^+},$ ${\leftarrow^+}$, ${\cup},{[\ ]})$,
the satisfiability problem for $X$ under MRW-DTDs is in P
if and only if
the same problem under MDF/DC-DTDs is in P.
\end{corollary}

\subsection{Schema graphs and sibling-constraint mappings}

First, we introduce schema graphs.
Let $D$ be an MDF/DC-DTD.
Let $\DDC$ denote the DC-DTD obtained by replacing
every disjunction operator appearing outside of any Kleene stars
in a content model with the concatenation operator.
For example, a content model $(a|b(c|d)^*)ef^*$
in $D$ is replaced with $(ab(c|d)^*)ef^*$ in $\DDC$.
Then, the schema graph of $D$ is defined as that of $\DDC$.

\begin{definition}
The \emph{schema graph}~\cite{IMSHF09} $G=(U,E)$ of
a DC-DTD $\DDC=(\Sigma,r,P)$ is a directed graph
defined as follows:
\begin{itemize}
\item
A node $u\in U$ is either
\begin{itemize}
\item
$(\bot,1,-,r)$,
where $\bot$ is a new symbol not in $\Sigma$,
or
\item
$(a,i,\omega,b)$, where
$a$, $b\in\Sigma$, $1\leq i\leq\len(P(a))$ such that
$b$ appears in the $i$-th subexpression $e_i$ of $P(a)$,
and $\omega=\mbox{``$-$''}$ if $e_i$ is a single symbol in $\Sigma$
and $\omega=\mbox{``$*$''}$ otherwise.
\end{itemize}
The first, second, third and fourth components of $u$ are
denoted by $\lambdapar(u)$, $\pos(u)$, $\omega(u)$, 
and $\lambda(u)$, respectively.
Especially, $\lambda(u)$ is called the \emph{label} of $u$.
$\lambdapar$, $\pos$, and $\lambda$ are extended to functions on
sequences.
\item
An edge from $u$ to $u'$ exists in $E$ if and only if
$\lambda(u)=\lambdapar(u')$.
\end{itemize}
The \emph{schema graph} of an MDF/DC-DTD $D$ is that of
the corresponding DC-DTD $\DDC$.
\end{definition}

\begin{example}
\label{ex:schema graph}
Let $D=(\{r,a,b,c\},r,P)$ be an MDF/DC-DTD, where
\[
P(r)=r^*(a^*b|c)r^*,\quad
P(a)=\epsilon,\quad
P(b)=a,\quad
P(c)=\epsilon.
\]
Then, the corresponding DC-DTD $\DDC=(\{r,a,b,c\},r,$ $\PDC)$ is
as follows:
\[
\PDC(r)=r^*a^*bcr^*,
\PDC(a)=\epsilon,
\PDC(b)=a,
\PDC(c)=\epsilon.
\]
The schema graph $G$ of $D$ and $\DDC$ is shown in
Figure~\ref{fig:schema graph}.
\end{example}

\begin{figure}[tb]
\centering{
\includegraphics[scale=0.3]{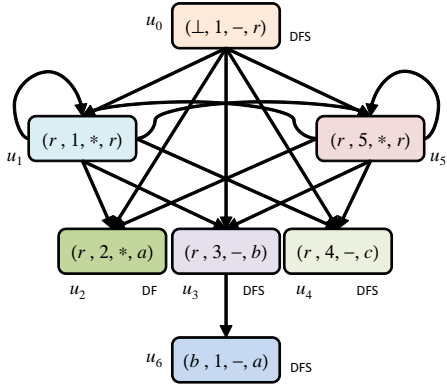}
}
\caption{Schema graph $G$.}
\label{fig:schema graph}
\end{figure}

Suppose that $\TDC\in\TL(\DDC)$ for a DC-DTD $\DDC$.
As stated in~\cite{IMSHF09},
there exists a mapping $\theta$,
called an \emph{SG mapping} of $\TDC$,
from the set of nodes of $\TDC$ to
the set of nodes of the schema graph of $\DDC$ with the following
properties:
\begin{itemize}
\item
$\theta$ maps the root node of $\TDC$ to $(\bot,1,-,r)$.
\item
Let $v$ be a node of $\TDC$ and
$v_1\cdots v_n$ be the children sequence of $v$.
Then, $\theta(v_j)=(\lambda(v),i_j,\omega_{i_j},\lambda(v_j))$, where
$1\leq i_j\leq\len(P(\lambda(v)))$,
$\omega_{i_j}=\mbox{``$-$''}$
if the $i_j$-th subexpression of $P(\lambda(v))$ is a single
symbol in $\Sigma$ and $\omega_{i_j}=\mbox{``$*$''}$ otherwise,
and $i_j\leq i_{j'}$ if $j\leq j'$.
Moreover, for every maximum subsequence $v_j\cdots v_{j'}$
such that $i_j=\cdots=i_{j'}$,
$\lambda(v_j\cdots v_{j'})$ is contained in $L(e_{i_j})$,
where $e_{i_j}$ is the $i_j$-th subexpression
of $P(\lambda(v))$.
\end{itemize}

Consider a tree $T\in\TL(D)$ for an MDF/DC-DTD $D$.
Then, there is a tree $\TDC\in\TL(\DDC)$ such that
$T$ is obtained by removing some subtrees of $\TDC$.
Hence, we can define an \emph{SG mapping} of $T$ as one of $\TDC$
whose domain is restricted to the nodes remaining in $T$.
SG mappings are extended to functions on sets and sequences
of nodes, i.e.,
$\theta(\{v_1,\ldots,v_n\})=\{\theta(v_1),\ldots,\theta(v_n)\}$ and
$\theta(v_1\cdots v_n)=\theta(v_1)\cdots\theta(v_n)$.

Let $e=e_1e_2\cdots e_n$ be an MDF/DC regular expression.
We say that each symbol in $e_i$ is \emph{DF} in $e$ if
$e_i$ is not DC.
Moreover, a DF symbol is \emph{DFS}
if it is outside the scope of $*$.
For example, consider $(a|b^*)cd^*$.
Then, $a$, $b$, and $c$ are DF but $d$ is not DF.
Also, $a$ and $c$ are DFS but $b$ is not DFS.
We define DF and DFS nodes of schema graphs in a similar way.
A path $s$ on a schema graph $G$ is \emph{DFS}
if $s$ consists of only DFS nodes of $G$
(we regard $(\bot,1,-,r)$ as DFS).

\begin{figure}[tb]
\centering{
\begin{minipage}{.2\textwidth}
\begin{center}
\includegraphics[scale=0.3]{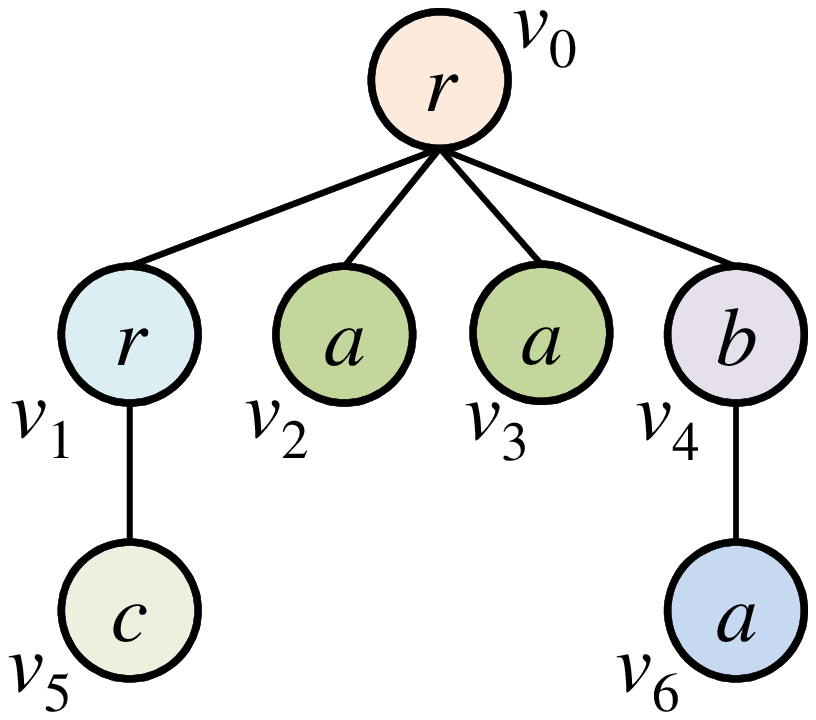}
\end{center}
\end{minipage}
\begin{minipage}{.25\textwidth}
\begin{eqnarray*}
\theta(v_0) & = & (\bot,1,-,r)
\\
\theta(v_1) & = & (r,1,*,r)
\\
\theta(v_2)=\theta(v_3) & = & (r,2,*,a)
\\
\theta(v_4) & = & (r,3,-,b)
\\
\theta(v_5) & = & (r,4,-,c)
\\
\theta(v_6) & = & (b,1,-,a)
\end{eqnarray*}
\end{minipage}
}
\caption{A tree $T$ and its SG mapping $\theta$.}
\label{fig:tree}
\end{figure}

\begin{example}
Consider the MDF/DC-DTD $D$ defined in Example~\ref{ex:schema graph},
and a tree $T\in\TL(D)$ shown in Figure~\ref{fig:tree}.
In this case, there is a unique SG mapping $\theta$ of $T$,
which is also shown in Figure~\ref{fig:tree}.
\end{example}

\begin{definition}
A \emph{sibling-constraint mapping} $\beta$
is a partial mapping from non-empty paths from
$(\bot,1,-,r)$ on $G=(U,E)$ to the powerset of
$U$ such that
\begin{enumerate}
\item
$\beta(s)$ is defined only for
a finite number of $s$;
and
\item
if defined,
$\beta(s)$ is a set of DF children of the last node of $s$.
\end{enumerate}
\end{definition}
We write $\beta\sqsupseteq\beta'$ if
$\beta(s)\supseteq\beta'(s)$ whenever
$\beta'(s)$ is defined.
Let $\beta\sqcup\beta'$ denote the least upper bound of
$\beta$ and $\beta'$ with respect to $\sqsupseteq$.

Let $T$ be a tree in $\TL(D)$ and $\theta$ be its SG mapping.
Let $\beta$ be a sibling-constraint mapping.
A pair $(T,\theta)$ \emph{satisfies} $\beta$ if
for each $s$ such that $\beta(s)$ is defined,
there is a path $w$ on $T$ such that $\theta(w)=s$ and
$\beta(s)\subseteq\theta(\SibDF_T(w))$, where
\begin{eqnarray*}
\SibDF_T(w\cdot v) & = & 
\{v'\mid\mbox{$v'$ is a child of $v$ in $T$}
\\
&&
\qquad\mbox{such that $\lambda(v')$ is DF in $P(\lambda(v))$}\}.
\end{eqnarray*}

\begin{definition}
A sibling-constraint mapping $\beta$ is \emph{consistent} if,
for each path $s\cdot u$ on $G$
such that $\beta(s\cdot u)$ is defined,
there exists a string in $L(P(\lambda(u)))$
that contains all $\lambda(u')$'s where
$u'\in\beta(s\cdot u)$.
\end{definition}

It is not difficult to see that
$\beta$ is consistent if and only if
there are a tree $T\in\TL(D)$ and its SG mapping $\theta$
such that $(T,\theta)$ satisfies $\beta$.

\begin{example}
\label{ex:sibling-constraint mapping}
Consider the schema graph $G$ in Figure~\ref{fig:schema graph}, and
let $\beta$ be the following sibling-constraint mapping:
$\beta(u_0)=\{u_2,u_3\}$,
$\beta(u_0u_1)=\{u_4\}$,
$\beta(u_0u_2)=\emptyset$, and
$\beta(u_0u_3)=\{u_6\}$.
Then, $\beta$ is consistent.
Actually, $(T,\theta)$ shown in Figure~\ref{fig:tree}
satisfies $\beta$.

Next, consider $\beta'=\beta\sqcup\{u_0\mapsto\{u_4\}\}$.
In this case, $\beta'(u_0)=\{u_2,u_3,u_4\}$ and
$\beta'$ is not consistent because
there is no string in $L(P(\lambda(u_0)))$
(i.e., $L(r^*(a^*b|c)r^*)$)
which contains
all of $\lambda(u_2)$, $\lambda(u_3)$, and $\lambda(u_4)$
(i.e., $a$, $b$, and $c$).
\end{example}

\subsection{A necessary and sufficient condition for XPath satisfiability}

We define a satisfaction relation $\modelsa$
between schema graphs and XPath expressions.
Then, we show that $\modelsa$ coincides with
XPath satisfiability under MDF/DC-DTDs.

In our previous work~\cite{ISF10},
we provided a satisfaction relation $\modelsb$
between schema graphs and XPath expressions,
and showed that $\modelsb$ coincides with
XPath satisfiability under DC-DTDs.
More precisely,
we showed that for any XPath expression
$p\in
\mathcal{X}({\downarrow},{\downarrow^\ast},{\uparrow},{\uparrow^\ast}$,
${\rightarrow^+},$ ${\leftarrow^+}$, ${\cup},{[~]})$,
$T\models p(w,w')$ if and only if
$G\modelsb p(\theta(w),\theta(w'))$,
where $\theta$ is an SG mapping of $T$.

Now, our target is MDF/DC-DTDs, so we have to analyze
non-cooccurrence specified by MDF/DC-DTDs.
To do so, we augment the parameters of $p$ by
sibling-constraint mappings introduced in the previous section.
That is, we will define $\modelsa$ so that,
roughly speaking, 
$G\modelsa p((\theta(w),\beta),(\theta(w'),\beta'))$
means that if $(T,\theta)$ satisfies $\beta$,
then $T$ satisfies $p$ at $w$ and $w'$ provided that
$(T,\theta)$ also satisfies $\beta'$.
In other words, $\beta$ is a pre-condition for $T$
before analyzing $p$,
and $\beta'$ is the post-condition for $T$
after analyzing $p$.

Actually, it is not necessary to keep
all sibling-constraint information.
Only the following cases must be handled by $\beta(s)$:
\begin{itemize}
\item
The case where $s$ is DFS.
Then, for any $T\in\TL(D)$,
$\lambda(s)$ is a unique label path on $T$ if exists.
Hence, the last node of the path can be visited many times.
So, sibling-constraint information $\beta(s)$
at $s$ must be maintained.
\item
The case where $s$ is a prefix of the ``current path''
of the analysis.
The last node of the ``current path'' can be considered as
the context node.
By using upward axes from the context node,
any ancestor node may be revisited.
So, sibling-constraint information $\beta(s)$
at such $s$ must be maintained.
\end{itemize}
On the other hand, if $s$ does not meet the two cases above,
$s$ contains a node inside the scope of some $*$.
There is no way to \emph{always} revisit the last node of $s$
in our XPath class,
sibling-constraint information $\beta(s)$
at such $s$ does not have to be maintained.

We provide the formal definition of $\modelsa$.
In what follows,
let $u$, $u'$, etc.\ be nodes of $G$, and
let $s$, $s'$, etc.\ be nonempty sequences of nodes of $G$
starting by $(\bot,1,-,r)$,
unless otherwise stated.
We introduce the following notations for readability:
\begin{eqnarray*}
\psi(u) & = &
\left\{
\begin{array}{ll}
\{u\} & \mbox{if $u$ is DF},
\\
\emptyset & \mbox{otherwise},
\end{array}
\right.
\\
\beta|_{\DFS,s}(s') & = &
\left\{
\begin{array}{ll}
\beta(s') & \mbox{if $s'$ is DFS or}
\\
& \mbox{\quad a proper prefix of $s$},
\\
\mbox{undefined} & \mbox{otherwise}.
\end{array}
\right.
\end{eqnarray*}

\begin{definition}
A satisfaction relation $\modelsa$ between a schema graph $G$
and an XPath expression
$p\in
\mathcal{X}({\downarrow},{\downarrow^\ast},{\uparrow},{\uparrow^\ast}$,
${\rightarrow^+},$ ${\leftarrow^+}$, ${\cup},{[~]})$
is defined as follows:
\begin{itemize}
\item
$G\modelsa(\downarrow::l)((s,\beta),(su',\beta'))$
if
\begin{itemize}
\item
path $su'$ exists in $G$,
\item
$\lambda(u')=l$,
\item
$\beta=\beta|_{\DFS,s}$,
\item
$\beta'=\beta\sqcup\{s\mapsto\psi(u')\}$, and
\item
both $\beta$ and $\beta'$ are consistent.
\end{itemize}
\item
$G\modelsa(\uparrow::l)((su,\beta),(s,\beta'))$
if
\begin{itemize}
\item
path $su$ exists in $G$,
\item
the label of the last node of $s$ is $l$,
\item
$\beta=\beta|_{\DFS,su}$,
\item
$\beta'=\beta|_{\DFS,s}$, and
\item
both $\beta$ and $\beta'$ are consistent.
\end{itemize}
\item
$G\modelsa(\downarrow^*::l)((s,\beta),(ss',\beta'))$
if
\begin{itemize}
\item
path $ss'$ exists in $G$,
where $s'$ is a possibly empty sequence of nodes of $G$,
\item
the label of the last node of $ss'$ is $l$,
\item
$\beta=\beta|_{\DFS,s}$,
\item
$\beta'=\beta\sqcup
\{s''\mapsto\psi(u')\mid\mbox{$s''u'$ is a prefix of $ss'$}\}$, and
\item
both $\beta$ and $\beta'$ are consistent.
\end{itemize}
\item
$G\modelsa(\uparrow^*::l)((ss',\beta),(s,\beta'))$
if
\begin{itemize}
\item
path $ss'$ exists in $G$,
where $s'$ is a possibly empty sequence of nodes of $G$,
\item
the label of the last node of $s$ is $l$,
\item
$\beta=\beta|_{\DFS,ss'}$,
\item
$\beta'=\beta|_{\DFS,s}$, and
\item
both $\beta$ and $\beta'$ are consistent.
\end{itemize}
\item
$G\modelsa(\rightarrow^+::l)((su,\beta),(su',\beta'))$ if
\begin{itemize}
\item
$\lambdapar(u)=\lambdapar(u')$,
\item
$\lambda(u')=l$,
\item
$\pos(u)<\pos(u')$ if $\omega(u)=\mbox{``$-$''}$ and
$\pos(u)\leq\pos(u')$ if $\omega(u)=\mbox{``$*$''}$,
\item
$\beta=\beta|_{\DFS,su}$,
\item
$\beta'=\beta\sqcup\{s\mapsto\psi(u')\}$, and
\item
both $\beta$ and $\beta'$ are consistent.
\end{itemize}
\item
$G\modelsa(\leftarrow^+::l)((su,\beta),(su',\beta'))$ if
\begin{itemize}
\item
$\lambdapar(u)=\lambdapar(u')$,
\item
$\lambda(u')=l$,
\item
$\pos(u)>\pos(u')$ if $\omega(u)=\mbox{``$-$''}$ and
$\pos(u)\geq\pos(u')$ if $\omega(u)=\mbox{``$*$''}$,
\item
$\beta=\beta|_{\DFS,su}$,
\item
$\beta'=\beta\sqcup\{s\mapsto\psi(u')\}$, and
\item
both $\beta$ and $\beta'$ are consistent.
\end{itemize}
\item
$G\modelsa(p/p')((s,\beta),(s',\beta'))$
if there is a pair $(s'',\beta'')$ such that
$G\modelsa p((s,\beta),(s'',\beta''))$ and
$G\modelsa p'((s'',\beta''),(s',\beta'))$.
\item
$G\modelsa(p\cup p')((s,\beta),(s',\beta'))$
if 
$G\modelsa p((s,\beta),(s',\beta'))$ or
$G\modelsa p'((s,\beta),(s',\beta'))$.
\item
$G\modelsa(p[q])((s,\beta),(s',\beta'\sqcup\beta'')$
if
$G\modelsa p((s,\beta),(s',\beta'))$,
$G\modelsa q((s',\beta''))$, and
$\beta'\sqcup\beta''$ is consistent.
\item
$G\modelsa p((s,\beta'|_{\DFS,s}))$
if there are $s'$, $\beta$, and $\beta'$ such that
$G\modelsa p((s,\beta),(s',\beta'))$.
\item
$G\modelsa (q\wedge q')((s,\beta\sqcup\beta'))$
if 
$G\modelsa q((s,\beta))$,
$G\modelsa q'((s,\beta'))$, and
$\beta\sqcup\beta'$ is consistent.
\item
$G\modelsa (q\vee q')((s,\beta))$
if 
$G\modelsa q((s,\beta))$ or
$G\modelsa q'((s,\beta))$.
\end{itemize}
\end{definition}

The following lemmas can be shown
immediately from the definition of $\modelsa$:

\begin{lemma}
\label{lem:compactness}
If $G\modelsa p((s,\beta),(s',\beta'))$,
then $\beta|_{\DFS,s}=\beta$ and
$\beta'|_{\DFS,s'}=\beta'$.
If $G\modelsa q((s,\beta))$,
then $\beta|_{\DFS,s}=\beta$.
\end{lemma}

\begin{lemma}
\label{lem:monotonicity}
Suppose that $G\modelsa p((s,\beta),(s',\beta'))$.
If $\beta(s'')$ is defined for a DFS path $s''$,
then $\beta(s'')\subseteq\beta'(s'')$.
\end{lemma}

Now, we show that XPath expression
$p\in
\mathcal{X}({\downarrow},{\downarrow^\ast},{\uparrow},{\uparrow^\ast},
{\rightarrow^+},$ ${\leftarrow^+}$, ${\cup},{[~]})$
is satisfiable under $D$
if and only if
$G\modelsa p(((\bot,1,-,r),\beta_\bot),(s',\beta'))$
for some $s'$ and $\beta'$,
where $\beta_\bot$ is a mapping undefined everywhere.
The following theorem corresponds to the only if part:
\begin{theorem}
\label{th:T=>G a}
Let $p\in
\mathcal{X}({\downarrow},{\downarrow^\ast},{\uparrow},{\uparrow^\ast},
{\rightarrow^+},$ ${\leftarrow^+}$, ${\cup},{[~]})$.
Let $D$ be an MDF/DC-DTD and $G$ be the schema graph of $D$.
\begin{enumerate}
\item
Suppose that
$T\models p(w,w')$ for some $T\in\TL(D)$ with an SG mapping $\theta$.
Let $\beta$ be an arbitrary mapping satisfied by $(T,\theta)$
such that $\beta=\beta|_{\DFS,\theta(w)}$.
Then, there is a mapping $\beta'$ satisfied by $(T,\theta)$
such that
$G\modelsa p((\theta(w),\beta),(\theta(w'),\beta'))$.
\item
Suppose that
$T\models q(w)$ for some $T\in\TL(D)$ with an SG mapping $\theta$.
Then, there is a mapping $\beta'$ satisfied by $(T,\theta)$
such that
$G\modelsa q((\theta(w),\beta'))$.
\end{enumerate}
\end{theorem}
\begin{proofs}
The theorem is proved by induction on the structure of $p$.

\emph{Basis.}
Suppose that $T\models (\downarrow::l)(w,wv')$
and that $(T,\theta)$ satisfies $\beta$.
If $\theta(w)$ is not DFS,
then $\beta(\theta(w))$ is undefined
since $\beta=\beta|_{\DFS,\theta(w)}$.
If $\theta(w)$ is DFS,
then $\beta(\theta(w))\cup\{\theta(v')\}$
does not violate the non-cooccurrence
because $(T,\theta)$ satisfies $\beta$ and
path $wv'$ exists in $T$.
Hence, $(T,\theta)$ also satisfies
$\beta'=\beta\sqcup\{\theta(w)\mapsto\psi(\theta(v'))\}$.
So, $G\modelsa(\downarrow::l)((\theta(w),\beta),(\theta(wv'),\beta'))$.

The other cases can be shown in a similar way.

\emph{Induction.}
Suppose that $T\models (p[q])(w,w')$
and that $(T,\theta)$ satisfies $\beta$.
By the definition of qualifiers,
$T\models p(w,w')$ and $T\models q(w')$.
Let $\beta$ be an arbitrary mapping satisfied by $(T,\theta)$
such that $\beta=\beta|_{\DFS,\theta(w)}$.
By inductive hypothesis,
there are mappings $\beta'$ and $\beta''$ satisfied by $(T,\theta)$
such that
$G\modelsa p((\theta(w),\beta),(\theta(w'),\beta'))$, and
$G\modelsa q((\theta(w'),\beta''))$.
Moreover, by Lemma~\ref{lem:compactness},
we have $\beta'=\beta'|_{\DFS,\theta(w')}$ and
$\beta''=\beta''|_{\DFS,\theta(w')}$,
and hence $\beta'\sqcup\beta''$ is satisfied by $(T,\theta)$.
This means that $\beta'\sqcup\beta''$ is consistent, and therefore,
$G\modelsa (p[q])
((\theta(w),\beta),(\theta(w'),\beta'\sqcup\beta''))$.

The other cases are similarly proved.
\end{proofs}

The if part is shown below:
\begin{theorem}
\label{th:G=>T a}
Let $p\in
\mathcal{X}({\downarrow},{\downarrow^\ast},{\uparrow},{\uparrow^\ast},
{\rightarrow^+},$ ${\leftarrow^+}$, ${\cup},{[~]})$.
Let $D$ be an MDF/DC-DTD.
\begin{enumerate}
\item
Suppose that $G\modelsa p((s,\beta),(s',\beta'))$.
Then, there are $T\in\TL(D)$, its SG mapping $\theta$,
and paths $w$ and $w'$ on $T$ such that
$\theta(w)=s$, $\theta(w')=s'$, 
$\beta$ and $\beta'$ are satisfied by $(T,\theta)$,
and $T\models p(w,w')$.
\item
Suppose that $G\modelsa q((s,\beta))$.
Then, there are $T\in\TL(D)$, its SG mapping $\theta$,
and path $w$ on $T$ such that
$\theta(w)=s$,
$\beta$ is satisfied by $(T,\theta)$,
and $T\models q(w)$.
\end{enumerate}
\end{theorem}
\begin{proofs}
Again, the theorem is proved by induction on the structure of $p$.

\emph{Basis.}
Suppose that 
$G\modelsa(\downarrow::l)((s,\beta),(su',\beta'))$.
Since $\beta'=\beta\sqcup\{s\mapsto\psi(u')\}$ is consistent,
there is a pair $(T,\theta)$ satisfying $\beta'$.
Moreover, since $\beta'(s)$ is defined, there is a path $w$ on $T$
such that $\theta(w)=s$ and $\beta'(s)\subseteq\theta(\SibDF_T(w))$.
If $u'$ is DF,
then $u'\in\beta'(s)$,
so $\SibDF_T(w)$ contains a node $v'$ such that $\theta(v')=u'$.
Otherwise, without destroying the properties of $T$ stated so far,
we can add a node $v'$ to $T$
as a child of the last node of $w$
so that $\theta(v')=u'$.
Hence, in both cases, there is a path $wv'$ on $T$
such that $T\models(\downarrow::l)(w,wv')$ and
$\theta(wv')=su'$.
Finally, since $\beta'\sqsupseteq\beta$, $\beta$ is also satisfied by
$(T,\theta)$.

Next, suppose that 
$G\modelsa(\uparrow::l)((su,\beta),(s,\beta'))$.
Since $\beta$ is consistent,
there is a pair $(T,\theta)$ satisfying $\beta$.
Moreover, since $\beta(s)$ is defined, there is a path $w$ on $T$
such that $\theta(w)=s$ and $\beta(s)\subseteq\SibDF_T(w)$.
If $u$ is DF, then $u\in\beta(s)$,
so $\SibDF_T(w)$ contains a node $v$ such that $\theta(v)=u$.
Otherwise, without destroying the properties of $T$ stated so far,
we can add a node $v$ to $T$
as a child of the last node of $w$
so that $\theta(v)=u$.
Hence, in both cases, there is a path $wv$ on $T$
such that $T\models(\uparrow::l)(wv,w)$ and
$\theta(wv)=su$.
Finally, since $\beta\sqsupseteq\beta'=\beta|_{\DFS,s}$,
$\beta'$ is also satisfied by
$(T,\theta)$.

The other cases can be shown in a similar way.

\emph{Induction.}
Suppose that 
$G\modelsa(p[q])((s,\beta),(s',\beta'\sqcup\beta''))$.
By the definition of $\modelsa$,
we have
$G\modelsa p((s,\beta),(s',\beta'))$ and
$G\modelsa q((s',\beta''))$.
By the inductive hypothesis,
\begin{itemize}
\item
there are $T_1\in\TL(D)$, its SG mapping $\theta_1$,
and paths $w_1$ and $w_1'$ on $T_1$ such that
$\theta_1(w_1)=s$,
$\theta_1(w_1')=s'$,
$\beta$ and $\beta'$ are satisfied by $(T_1,\theta_1)$,
and $T_1\models p(w_1,w_1')$;
and
\item
there are $T_2\in\TL(D)$, its SG mapping $\theta_2$,
and path $w_2$ on $T_2$ such that
$\theta_2(w_2)=s'$,
$\beta''$ is satisfied by $(T_2,\theta_2)$,
and $T_2\models q(w_2)$.
\end{itemize}

Let $T\in\TL(D)$ be the tree obtained by 
merging $T_1$ and $T_2$ so that
DFS paths of $T_1$ and $T_2$ are overlapped and
$w_1'$ and $w_2$ are also overlapped.
This is possible because Lemma~\ref{lem:monotonicity} holds,
$\beta'\sqcup\beta''$ is consistent,
and $\theta_1(w_1')=\theta_2(w_2)=s'$.
An SG mapping $\theta$ of $T$
can be defined as an extension of
both $\theta_1$ and $\theta_2$.
Hence, $\theta(w_1)=\theta_1(w_1)=s$,
$\theta(w_1')=\theta_1(w_1')=\theta_2(w_2)=s'$,
$\beta$ and $\beta'\sqcup\beta''$ are satisfied by $(T,\theta)$,
and $T\models (p[q])(w_1,w_1')$.


The other cases are similarly proved.
\end{proofs}

\subsection{Tractability}

In this section,
we show that the necessary and sufficient condition is
decidable in polynomial time if
$p\in\mathcal{X}({\downarrow},{\uparrow},
{\rightarrow^+},{\leftarrow^+})$ or
$p\in\mathcal{X}({\downarrow},
{\rightarrow^+},{\leftarrow^+},[~]_\wedge)$.



\subsubsection{$\mathcal{X}({\downarrow},{\uparrow},
{\rightarrow^+},{\leftarrow^+})$}

Let $p\in\mathcal{X}({\downarrow},{\uparrow},
{\rightarrow^+},{\leftarrow^+})$.
We show an efficient algorithm for
deciding whether $G\modelsa p(((\bot,1,-,r),\beta_\bot),(s',\beta'))$
for some $s'$ and $\beta'$.

Essentially, our algorithm $\eval_1$ runs in a top-down manner
with respect to the parse tree of $p$,
and computes the set of the second parameters
$(s',\beta')$ of $p$
for a given set of first parameters $(s,\beta)$.
Let $B$ denote a set of pairs of
a path on $G$ and a sibling-constraint mapping.
Formally, $\eval_1$ is defined as follows:
\[
\eval_1(p,B)=
\left\{
\begin{array}{ll}
\multicolumn{2}{l}{
\{(s',\beta')\mid\mbox{$G\modelsa p((s,\beta),(s',\beta'))$}
}
\\
\qquad\quad\mbox{for each $(s,\beta)\in B$}\}
&
\mbox{if $p$ is atomic},
\\
\eval_1(p_2,\eval_1(p_1,B))
&
\mbox{if $p=p_1/p_2$}.
\end{array}
\right.
\]
In what follows,
we show that $\eval_1(p,\{((\bot,1,-,r),\beta_\bot)\})$
runs in a polynomial time.

First, given $s$, $s'$, and $\beta$,
there is at most one $\beta'$ such that
$G\modelsa p((s,\beta)$, $(s',\beta'))$.
That is, the combination of $s$ and $\beta$
does not cause combinatorial explosion.
This property is formally stated by the following lemma:
\begin{lemma}
\label{lem:uniqueness 1a}
Let $p\in\mathcal{X}({\downarrow},{\uparrow},
{\rightarrow^+},{\leftarrow^+})$.
Suppose that
$G\modelsa p((s,\beta)$, $(s',\beta'))$ and
$G\modelsa p((s,\beta)$, $(s',\beta''))$.
Then, $\beta'=\beta''$.
\end{lemma}
\begin{proof}
Immediate from the definition of $\modelsa$ since
$p$ contains none of
$\downarrow^*$, $\uparrow^*$, $\cup$, and $[~]$.
\end{proof}

Next, consider the explosion of the number of $s$.
Because of the nondeterminism of $\downarrow$,
$\rightarrow^+$, and $\leftarrow^+$,
the number of $s$ can be exponential in the size of $p$.
However, recall that we are interested in $(s',\beta')$ such that
$G\modelsa p(((\bot,1,-,r),\beta_\bot),(s',\beta'))$.
The following lemma implies that
such $s'$ is unique up to the labeling function $\lambda$.
Moreover, $\beta'$ is also unique up to $\lambda$:
\begin{lemma}
\label{lem:uniqueness 1b}
Let $p\in\mathcal{X}({\downarrow},{\uparrow},
{\rightarrow^+},{\leftarrow^+})$.
Suppose that
$G\modelsa p((s_1,\beta_1)$, $(s_1',\beta_1'))$ and
$G\modelsa p((s_2,\beta_2)$, $(s_2',\beta_2'))$,
where $\lambda(s_1)=\lambda(s_2)$ and
$\beta_1(s_1'')=\beta_2(s_2'')$ for all $s_1''$ and $s_2''$
such that $\lambda(s_1'')=\lambda(s_2'')$.
Then, $\lambda(s_1')=\lambda(s_2')$ and
$\beta_1(s_1'')=\beta_2(s_2'')$ for all $s_1''$ and $s_2''$
such that $\lambda(s_1'')=\lambda(s_2'')$.
\end{lemma}
\begin{proof}
The lemma can be shown by induction on the structure of $p$.
\end{proof}
Let $\beta'/_\lambda$ denote the mapping such that
$\beta'/_\lambda(\lambda(s''))=\beta'(s'')$ for any $s''$.
Operators $\sqcup$ and $|_{\DFS,s}$ and consistency
can be naturally redefined on $\beta'/_\lambda$ as long as
$\beta'=\beta'|_{\DFS,s}$.
We have to maintain only one $\beta'/_\lambda$
even if the number of $s'$ explodes.

Finally, we have to introduce a concise representation
of exponentially many $s'$.
To accomplish this, the following observation is useful:
\begin{lemma}
\label{lem:independence 1}
Let $p\in\mathcal{X}({\downarrow},{\uparrow},
{\rightarrow^+},{\leftarrow^+})$ be an atomic XPath expression.
Suppose that
$G\modelsa p((su,\beta/_\lambda)$, $(ss',\beta'/_\lambda))$.
Then, for any path $s''$ on $G$ such that $\lambda(s'')=\lambda(s)$,
we have $G\modelsa p((s''u,\beta/_\lambda)$, $(s''s',\beta'/_\lambda))$.
\end{lemma}
\begin{proof}
Immediate from the definition of $\modelsa$ since
$p$ contains neither
$\downarrow^*$ nor $\uparrow^*$.
\end{proof}
In other words, for atomic $p$,
only the last node of $s$ is meaningful.
Hence, we use a sequence $U_0U_1\cdots U_n$ of
sets of nodes of $G$ for representing the set of $s$ or $s'$,
where $U_0=\{(\bot,1,-,r)\}$.
As usual,
$s=u_0u_1\cdots u_n$ is in $U_0U_1\cdots U_n$
if $u_i\in U_i$ for each $i$.

The following is a refined version of our algorithm $\eval_1$:

\medskip\par\noindent
$\eval_1(p,(U_0\cdots U_n,\beta/_\lambda)):$
\begin{itemize}
\item
If $p={\downarrow}::l$, then return
\[
(U_0\cdots U_nU_{n+1},\beta/_\lambda\sqcup\{s\mapsto\psi(u')\}/_\lambda),
\]
where $s$ is an arbitrary path in $U_0\cdots U_n$,
$u'$ is an arbitrary node such that $su'$ is a path on $G$
and the label of $u'$ is $l$,
and $U_{n+1}$ is the set of such nodes $u'$.
If $\beta/_\lambda\sqcup\{s\mapsto\psi(u')\}/_\lambda$ is not consistent, 
then the execution of $\eval_1$ fails
(i.e., $p$ is unsatisfiable).
\item
If $p={\uparrow}::l$, then return
\[
(U_0\cdots U_{n-1},\beta/_\lambda|_{\DFS,\lambda(s)}),
\]
where $s$ is an arbitrary path in $U_0\cdots U_{n-1}$
such that the label of the last node of $s$ is $l$.
\item
If $p={\rightarrow^+}::l$, then return
\[
(U_0\cdots U_{n-1}U_n',\beta/_\lambda\sqcup\{s\mapsto\psi(u')\}/_\lambda),
\]
where $s$ is an arbitrary path in $U_0\cdots U_{n-1}$,
$u'$ is an arbitrary node such that $su'$ is a path on $G$,
the label of $u'$ is $l$,
and there is $u\in U_n$ such that
$\pos(u)<\pos(u')$ if $\omega(u)=\mbox{``$-$''}$ and
$\pos(u)\leq\pos(u')$ if $\omega(u)=\mbox{``$*$''}$,
and $U_n'$ is the set of such nodes $u'$.
If $\beta/_\lambda\sqcup\{s\mapsto\psi(u')\}/_\lambda$ is not consistent, 
then the execution of $\eval_1$ fails.
The case of $p={\leftarrow^+}::l$ is similar.
\item
If $p=p_1/p_2$, then return
\[
\eval_1(p_2,\eval_1(p_1,(U_0\cdots U_n,\beta/_\lambda))).
\]
\end{itemize}

Let $G=(U,E)$.
It takes $O(|U|)$ time to process an atomic XPath expression.
Totally, it takes $O(|p||U|)$ time
to run $\eval_1(p,(\{(\bot,1,-,r)\},\beta_\bot/_\lambda))$.

\begin{theorem}
XPath satisfiability for
$\mathcal{X}({\downarrow},{\uparrow},
{\rightarrow^+},{\leftarrow^+})$
under MRW-DTDs is decidable in polynomial time.
\end{theorem}

\begin{example}
Let $D$ be the MDF/DC-DTD given in Example~\ref{ex:schema graph}.
Consider the satisfiability of
$p=({\downarrow}::r/{\rightarrow^+}::b)/({\downarrow}::a/{\uparrow}::b)$
under $D$.
The execution of $\eval_1$ is as follows.
Recall that the schema graph of $D$ is given in
Figure~\ref{fig:schema graph}.
\begin{eqnarray*}
\lefteqn{\eval_1(p,(\{u_0\},\beta_\bot/_\lambda))}
\\
& = &
\eval_1({\downarrow}::a/{\uparrow}::b,
\eval_1({\downarrow}::r/{\rightarrow^+}::b,(\{u_0\},\beta_\bot/_\lambda)))
\\
& = &
\eval_1({\downarrow}::a/{\uparrow}::b,
\\
&&
\qquad\qquad
\eval_1({\rightarrow^+}::b,
\eval_1({\downarrow}::r,(\{u_0\},\beta_\bot/_\lambda))))
\\
& = &
\eval_1({\downarrow}::a/{\uparrow}::b,
\\
&&
\qquad\qquad
\eval_1({\rightarrow^+}::b,
(\{u_0\}\{u_1,u_5\},\{r\mapsto\emptyset\})))
\\
& = &
\eval_1({\downarrow}::a/{\uparrow}::b,
(\{u_0\}\{u_3\},\{r\mapsto\{b\}\}))
\\
& = &
\eval_1({\uparrow}::b,\eval_1({\downarrow}::a,
(\{u_0\}\{u_3\},\{r\mapsto\{b\}\})))
\\
& = &
\eval_1({\uparrow}::b,
(\{u_0\}\{u_3\}\{u_6\},\{r\mapsto\{b\},rb\mapsto\{a\}\}))
\\
& = &
(\{u_0\}\{u_3\},\{r\mapsto\{b\},rb\mapsto\{a\}\}).
\end{eqnarray*}
Hence $p$ is determined to be satisfiable.
Actually, the tree $T$ in Figure~\ref{fig:tree} satisfies $p$.

Next, consider the satisfiability $p'=p/{\rightarrow^+}::c$
under $D$.
Then, the execution of $\eval_1$ would be as follows:
\begin{eqnarray*}
\lefteqn{\eval_1(p',(\{u_0\},\beta_\bot/_\lambda))}
\\
& = &
\eval_1({\rightarrow^+}::c,
\eval_1(p,(\{u_0\},\beta_\bot/_\lambda)))
\\
& = &
\eval_1({\rightarrow^+}::c,(\{u_0\}\{u_3\},\{r\mapsto\{b\},rb\mapsto\{a\}\}))
\\
& = &
(\{u_0\}\{u_4\},\{r\mapsto\{b,c\},rb\mapsto\{a\}\}).
\end{eqnarray*}
However, $\{r\mapsto\{b,c\},rb\mapsto\{a\}\}$ is not consistent,
and hence $p'$ is determined to be unsatisfiable.
\end{example}

\subsubsection{$\mathcal{X}({\downarrow},
{\rightarrow^+},{\leftarrow^+},[~]_\wedge)$}

Let $p\in\mathcal{X}({\downarrow},
{\rightarrow^+},{\leftarrow^+},[~]_\wedge)$.
We show an efficient algorithm for
deciding whether $G\modelsa p(((\bot,1,-,r),\beta_\bot),(s',\beta'))$
for some $s'$ and $\beta'$.

For this case, our algorithm $\eval_2$ runs in a bottom-up manner
with respect to the parse tree of $p$,
and essentially computes the set of all the pairs
$((s,\beta),(s',\beta'))$ such that
$G\modelsa p((s,\beta),(s',\beta'))$.
However, a naive implementation causes
exponential runtime.
Since the same properties as Lemmas~\ref{lem:uniqueness 1a} and
\ref{lem:uniqueness 1b} hold for this XPath class,
we can use the ideas again in the previous section.
Moreover, since this XPath class contains no upward axes,
it suffices to maintain just the last nodes of $s$ and $s'$.
However, to handle path concatenations and qualifiers,
we need information which parameter of $\beta'$ is the
``current node,'' which is originally represented by $s'$.
Here, we use $\lambda(s')$ instead of $s'$ itself
to avoid explosion.
To summarize,
let us allow arbitrary (possibly empty) paths on $G$
as parameters of sibling-constraint mappings $\beta$, and
let $\beta\oslash s''$ denote a mapping such that
$(\beta\oslash s'')(s''s)=\beta(s)$.
Now, $\eval_2$ computes
all the tuples $((u,\beta/_\lambda),(u',\beta'/_\lambda),\lambda(s'))$
such that
$G\modelsa p((s''u,\beta\oslash s''),(s''s'u',\beta'\oslash s''))$
for any $s''$, where $\beta$ is the minimum mapping with respect to
$\sqsupseteq$.

The following is a formal description of our algorithm $\eval_2$:

\medskip\par\noindent
$\eval_2(p):$
\begin{itemize}
\item
If $p={\downarrow}::l$, then return the set of
\[
((u,\beta_\bot/_\lambda),
(u',\{u\mapsto\psi(u')\}/_\lambda),\lambda(u)),
\]
where $u\in U$,
$uu'$ is a path on $G$,
and the label of $u'$ is $l$.
\item
If $p={\rightarrow^+}::l$, then return the set of
\[
((u,\{\epsilon\mapsto\psi(u)\}/_\lambda),
(u',\{\epsilon\mapsto\psi(u)\cup\psi(u')\}/_\lambda),\epsilon),
\]
where $u\in U$,
$u'$ is a sibling node of $u$,
the label of $u'$ is $l$,
$\pos(u)<\pos(u')$ if $\omega(u)=\mbox{``$-$''}$ and
$\pos(u)\leq\pos(u')$ if $\omega(u)=\mbox{``$*$''}$, and
$\{\epsilon\mapsto\psi(u)\cup\psi(u')\}/_\lambda$ is consistent.
The case of $p={\leftarrow^+}::l$ is similar.
\item
If $p=p_1/p_2$, then return the set of
\[
((u_1,\beta_1/_\lambda),
(u_2,\beta_1'/_\lambda\sqcup(\beta_2'/_\lambda\oslash x_1)),x_1x_2),
\]
such that
\begin{eqnarray*}
((u_1,\beta_1/_\lambda),(u,\beta_1'/_\lambda),x_1)
& \in & \eval_2(p_1),
\\
((u,\beta_2/_\lambda),(u_2,\beta_2'/_\lambda),x_2)
& \in & \eval_2(p_2),
\end{eqnarray*}
and $(\beta_1'/_\lambda\sqcup(\beta_2'/_\lambda\oslash x_1))$
is consistent, where
\[
(\beta_2'/_\lambda\oslash x_1)(x_1x_2)=\beta_2'/_\lambda(x_2).
\]
\item
If $p=p_1[p_2]$, then return the set of
\[
((u_1,\beta_1/_\lambda),
(u,(\beta_1'/_\lambda\sqcup(\beta_2'/_\lambda\oslash x_1))|_{\DFS,x_1}),x_1),
\]
such that
\begin{eqnarray*}
((u_1,\beta_1/_\lambda),(u,\beta_1'/_\lambda),x_1)
& \in & \eval_2(p_1),
\\
((u,\beta_2/_\lambda),(u_2,\beta_2'/_\lambda),x_2)
& \in & \eval_2(p_2),
\end{eqnarray*}
and $(\beta_1'/_\lambda\sqcup(\beta_2'/_\lambda\oslash x_1))|_{\DFS,x_1}$
is consistent.
\end{itemize}

It takes $O(|U|^2)$ time to process an atomic XPath expression.
The number of output pairs for each subexpression is also $O(|U|^2)$.
Totally, it takes $O(|p||U|^4)$ time to run $\eval_2(p)$.

\begin{theorem}
XPath satisfiability for
$\mathcal{X}({\downarrow},
{\rightarrow^+},{\leftarrow^+},[~]_\wedge)$
under MRW-DTDs is decidable in polynomial time.
\end{theorem}

\begin{example}
Let $D$ be the MDF/DC-DTD given in Example~\ref{ex:schema graph}.
Consider the satisfiability of
$p={\downarrow}::r/{\rightarrow^+}::b[{\downarrow}::a]$
under $D$.
The execution of $\eval_2$ is as follows:
\begin{eqnarray*}
\eval_2({\downarrow}::r)
& = &
\{
((u_0,\beta_\bot/_\lambda),(u_1,\{r\mapsto\emptyset\}),r),
\\
&&\phantom{\{}
((u_0,\beta_\bot/_\lambda),(u_5,\{r\mapsto\emptyset\}),r),
\\
&&\phantom{\{}
((u_1,\beta_\bot/_\lambda),(u_1,\{r\mapsto\emptyset\}),r),
\\
&&\phantom{\{}
((u_1,\beta_\bot/_\lambda),(u_5,\{r\mapsto\emptyset\}),r),
\\
&&\phantom{\{}
((u_5,\beta_\bot/_\lambda),(u_1,\{r\mapsto\emptyset\}),r),
\\
&&\phantom{\{}
((u_5,\beta_\bot/_\lambda),(u_5,\{r\mapsto\emptyset\}),r)
\},
\\
\eval_2({\rightarrow^+}::b)
& = &
\{
((u_1,\{\epsilon\mapsto\emptyset\}),(u_3,\{\epsilon\mapsto\{b\}\}),\epsilon),
\\
&&\phantom{\{}
((u_2,\{\epsilon\mapsto\{a\}\}),
\\
&&\qquad\qquad\qquad
(u_3,\{\epsilon\mapsto\{a,b\}\}),\epsilon)
\},
\\
\eval_2({\downarrow}::a)
& = &
\{
((u_0,\beta_\bot/_\lambda),(u_2,\{r\mapsto\{a\}\}),r),
\\
&&\phantom{\{}
((u_1,\beta_\bot/_\lambda),(u_2,\{r\mapsto\{a\}\}),r),
\\
&&\phantom{\{}
((u_5,\beta_\bot/_\lambda),(u_2,\{r\mapsto\{a\}\}),r),
\\
&&\phantom{\{}
((u_3,\beta_\bot/_\lambda),(u_6,\{b\mapsto\{a\}\}),b)
\},
\end{eqnarray*}
\begin{eqnarray*}
\lefteqn{\eval_2({\rightarrow^+}::b[{\downarrow}::a])}
\\
& = &
\{
((u_1,\{\epsilon\mapsto\emptyset\}),
(u_3,\{\epsilon\mapsto\{b\},b\mapsto\{a\}\}),\epsilon),
\\
&&\phantom{\{}
((u_2,\{\epsilon\mapsto\{a\}\}),
(u_3,\{\epsilon\mapsto\{a,b\},b\mapsto\{a\}\}),\epsilon)
\},
\\
\lefteqn{\eval_2({\downarrow}::r/{\rightarrow^+}::b[{\downarrow}::a])}
\\
& = &
\{
((u_0,\beta_\bot/_\lambda),
(u_3,\{r\mapsto\{b\},rb\mapsto\{a\}\}),r),
\\
&&\phantom{\{}
((u_1,\beta_\bot/_\lambda),
(u_3,\{r\mapsto\{b\},rb\mapsto\{a\}\}),r),
\\
&&\phantom{\{}
((u_5,\beta_\bot/_\lambda),
(u_3,\{r\mapsto\{b\},rb\mapsto\{a\}\}),r)
\}.
\end{eqnarray*}
Since we have found $(s',\beta')$ such that
$G\modelsa p((u_0,\beta_\bot),(s',\beta'))$,
$p$ is determined to be satisfiable.
Actually, the tree $T$ in Figure~\ref{fig:tree} satisfies $p$.

Next, consider the satisfiability $p'=p/{\rightarrow^+}::c$
under $D$.
Then, the execution of $\eval_2({\rightarrow^+}::c)$ is:
\begin{eqnarray*}
\lefteqn{\eval_2({\rightarrow^+}::c)}
\\
& = &
\{
((u_1,\{\epsilon\mapsto\emptyset\}),(u_4,\{\epsilon\mapsto\{c\}\}),\epsilon),
\\
&&\phantom{\{}
((u_2,\{\epsilon\mapsto\{a\}\}),(u_4,\{\epsilon\mapsto\{a,c\}\}),\epsilon),
\\
&&\phantom{\{}
((u_3,\{\epsilon\mapsto\{b\}\}),(u_4,\{\epsilon\mapsto\{b,c\}\}),\epsilon)
\}.
\end{eqnarray*}
Note that $\{\epsilon\mapsto\{a,c\}\}$ and
$\{\epsilon\mapsto\{b,c\}\}$ are consistent
because they are undefined at non-empty paths.
Finally, the execution of $\eval_2(p/{\rightarrow^+}::c)$ would be
\begin{eqnarray*}
&&
\{
((u_0,\beta_\bot/_\lambda),
(u_4,\{r\mapsto\{b,c\},rb\mapsto\{a\}\}),r),
\\
&&\phantom{\{}
((u_1,\beta_\bot/_\lambda),
(u_4,\{r\mapsto\{b,c\},rb\mapsto\{a\}\}),r),
\\
&&\phantom{\{}
((u_5,\beta_\bot/_\lambda),
(u_4,\{r\mapsto\{b,c\},rb\mapsto\{a\}\}),r)
\}.
\end{eqnarray*}
However, since $\{r\mapsto\{b,c\},rb\mapsto\{a\}\}$ is not consistent,
$\eval_2(p')$ returns the empty set.
Hence $p'$ is determined to be unsatisfiable.
\end{example}


\section{Conclusions}
\label{sec:Conclusions}

This paper has proposed a class of DTDs,
called MRW-DTDs,
which cover many of the real-world DTDs
and have non-trivial tractability of XPath satisfiability.
To be specific,
MRW-DTDs cover
24 out of the 27 real-world DTDs, 1403 out of the 1407 DTD rules.
Under MRW-DTDs, we have shown that
satisfiability problems for
$\mathcal{X}({\downarrow},{\uparrow},{\rightarrow^+},
{\leftarrow^+})$ and
$\mathcal{X}({\downarrow},{\rightarrow^+},
{\leftarrow^+},[~]_\wedge)$ are both tractable.

Actually, we tried to show the intractability of
the union
$\mathcal{X}({\downarrow},{\uparrow},{\rightarrow^+},
{\leftarrow^+},[~]_\wedge)$ of the tractable classes.
However, we have finally found that reduction from 3SAT
to the class is very difficult.
One of our future work is to develop an efficient algorithm
for determining the satisfiability of the union class
under MRW-DTDs.

As stated in Section~\ref{sec:Introduction},
there have been two approaches to resolving the intractability of
XPath satisfiability.
The approach using fast decision procedures for MSO and
$\mu$-calculus is fairly powerful from the practical point of view.
It is reported~\cite{GL06,GL07,GLS07}
that satisfiability (and other static analysis problems
such as containment and coverage) was decided
within one second
for many XPath expressions taken from XPathMark~\cite{Fr05}.
Another important direction of the future work is empirical evaluation
of the proposed polynomial-time algorithms.



\subsection*{Acknowledgment}

The authors thank the anonymous reviewers for their insightful
and constructive comments and suggestions.
This research is supported in part by Grant-in-Aid for Scientific
Research~(C) 23500120 from Japan Society for the Promotion of Science.



\begin{thebibliography}{16}
\providecommand{\natexlab}[1]{#1}
\providecommand{\url}[1]{\texttt{#1}}
\expandafter\ifx\csname urlstyle\endcsname\relax
  \providecommand{\doi}[1]{doi: #1}\else
  \providecommand{\doi}{doi: \begingroup \urlstyle{rm}\Url}\fi

\bibitem[Arenas et~al.(2010)Arenas, Barcelo, Libkin, and Murlak]{ABLM10}
M.~Arenas, P.~Barcelo, L.~Libkin, and F.~Murlak.
\newblock \emph{Relational and {XML} Data Exchange}.
\newblock Morgan \& Claypool, 2010.

\bibitem[Benedikt et~al.(2005)Benedikt, Fan, and Geerts]{BFG05}
M.~Benedikt, W.~Fan, and F.~Geerts.
\newblock {XPath} satisfiability in the presence of {DTD}s.
\newblock In \emph{Proceedings of the Twenty-fourth ACM SIGACT-SIGMOD-SIGART
  Symposium on Principles of Database Systems}, pages 25--36, 2005.

\bibitem[Benedikt et~al.(2008)Benedikt, Fan, and Geerts]{BFG08}
M.~Benedikt, W.~Fan, and F.~Geerts.
\newblock {XPath} satisfiability in the presence of {DTD}s.
\newblock \emph{Journal of the ACM}, 55\penalty0 (2), 2008.

\bibitem[Franceschet(2005)]{Fr05}
M.~Franceschet.
\newblock {XPathMark}: An {XPath} benchmark for the {XMark} generated data.
\newblock In \emph{Proceedings of the Third International XML Database
  Symposium}, pages 129--143, 2005.

\bibitem[Geerts and Fan(2005)]{GF05}
F.~Geerts and W.~Fan.
\newblock Satisfiability of {XPath} queries with sibling axes.
\newblock In \emph{Proceedings of the 10th International Symposium on Database
  Programming Languages}, pages 122--137, 2005.

\bibitem[Genev{\`e}s and Laya\"{\i}da(2006)]{GL06}
P.~Genev{\`e}s and N.~Laya\"{\i}da.
\newblock A system for the static analysis of {XPath}.
\newblock \emph{ACM Transactions on Information Systems}, 24\penalty0
  (4):\penalty0 475--502, 2006.

\bibitem[Genev{\`e}s and Laya\"{\i}da(2007)]{GL07}
P.~Genev{\`e}s and N.~Laya\"{\i}da.
\newblock Deciding {XPath} containment with {MSO}.
\newblock \emph{Data \& Knowledge Engineering}, 63\penalty0 (1):\penalty0
  108--136, 2007.

\bibitem[Genev{\`e}s et~al.(2007)Genev{\`e}s, Laya\"{\i}da, and Schmitt]{GLS07}
P.~Genev{\`e}s, N.~Laya\"{\i}da, and A.~Schmitt.
\newblock Efficient static analysis of {XML} paths and types.
\newblock In \emph{Proceedings of the ACM SIGPLAN 2007 Conference on
  Programming Language Design and Implementation}, pages 342--351, 2007.

\bibitem[Ishihara et~al.(2009)Ishihara, Morimoto, Shimizu, Hashimoto, and
  Fujiwara]{IMSHF09}
Y.~Ishihara, T.~Morimoto, S.~Shimizu, K.~Hashimoto, and T.~Fujiwara.
\newblock A tractable subclass of {DTD}s for {XPath} satisfiability with
  sibling axes.
\newblock In \emph{Proceedings of the 12th International Symposium on Database
  Programming Languages}, pages 68--83, 2009.

\bibitem[Ishihara et~al.(2010)Ishihara, Shimizu, and Fujiwara]{ISF10}
Y.~Ishihara, S.~Shimizu, and T.~Fujiwara.
\newblock Extending the tractability results on {XPath} satisfiability with
  sibling axes.
\newblock In \emph{Proceedings of the 7th International XML Database
  Symposium}, pages 33--47, 2010.

\bibitem[Ishihara et~al.(2012)Ishihara, Hashimoto, Shimizu, and
  Fujiwara]{IHSF12}
Y.~Ishihara, K.~Hashimoto, S.~Shimizu, and T.~Fujiwara.
\newblock {XPath} satisfiability with downward and sibling axes is tractable
  under most of real-world {DTD}s.
\newblock In \emph{Proceedings of the 12th International Workshop on Web
  Information and Data Management}, pages 11--18, 2012.

\bibitem[Kuwada et~al.(2013)Kuwada, Hashimoto, Ishihara, and Fujiwara]{KHIF13}
H.~Kuwada, K.~Hashimoto, Y.~Ishihara, and T.~Fujiwara.
\newblock The consistency and absolute consistency problems of {XML} schema
  mappings between restricted {DTD}s.
\newblock In \emph{Proceedings of the 15th Asia-Pacific Web Conference, LNCS
  7808}, pages 228--239, 2013.

\bibitem[Lakshmanan et~al.(2004)Lakshmanan, Ramesh, Wang, and Zhao]{LRWZ04}
L.~V.~S. Lakshmanan, G.~Ramesh, H.~Wang, and Z.~J. Zhao.
\newblock On testing satisfiability of tree pattern queries.
\newblock In \emph{Proceedings of the Thirtieth International Conference on
  Very Large Data Bases}, pages 120--131, 2004.

\bibitem[Montazerian et~al.(2007)Montazerian, Wood, and Mousavi]{MWM07}
M.~Montazerian, P.~T. Wood, and S.~R. Mousavi.
\newblock {XPath} query satisfiability is in {PTIME} for real-world {DTD}s.
\newblock In \emph{Proceedings of the 5th International XML Database Symposium,
  LNCS 4704}, pages 17--30, 2007.

\bibitem[Murata et~al.(2005)Murata, Lee, Mani, and Kawaguchi]{MLMK05}
M.~Murata, D.~Lee, M.~Mani, and K.~Kawaguchi.
\newblock Taxonomy of {XML} schema languages using formal language theory.
\newblock \emph{ACM Transactions on Internet Technology}, 5\penalty0
  (4):\penalty0 660--704, 2005.

\bibitem[Suzuki and Fukushima(2009)]{SF09}
N.~Suzuki and Y.~Fukushima.
\newblock Satisfiability of simple {XPath} fragments in the presence of {DTD}.
\newblock In \emph{Proceedings of the 11th International Workshop on Web
  Information and Data Management}, pages 15--22, 2009.

\end{thebibliography}

\end{document}